\documentclass[11pt,onecolumn,superscriptaddress,floatfix,showpacs]{revtex4-1}
\usepackage[utf8]{inputenc}  
\usepackage[T1]{fontenc}     
\usepackage[british]{babel}  
\usepackage[sc,osf]{mathpazo}
\usepackage[scaled=0.86]{berasans}  
\usepackage[colorlinks=true, citecolor=blue, urlcolor=blue]{hyperref}  
\usepackage{graphicx} 
\usepackage[babel]{microtype}  
\usepackage{amsmath,amssymb,amsthm,bm,amsfonts,mathrsfs,bbm} 
\usepackage{csquotes}

\usepackage{xspace}  
\usepackage{pgfplots}

\newcommand\id{\leavevmode\hbox{\small1\kern-3.3pt\normalsize1}}

\newtheorem{theorem}{Theorem}

\begin{document}
\title{Characterization of Nonlocal Resources Under Global Unitary Action}

\author{Arup Roy}
\email{arup145.roy@gmail.com}
\affiliation{Physics and Applied Mathematics Unit, Indian Statistical Institute, 203 B. T. Road, Kolkata 700108, India.}

\author{Some Sankar Bhattacharya}
\email{somesankar@gmail.com}
\affiliation{Physics and Applied Mathematics Unit, Indian Statistical Institute, 203 B. T. Road, Kolkata 700108, India.}

\author{Amit Mukherjee}
\email{amitisiphys@gmail.com}
\affiliation{Physics and Applied Mathematics Unit, Indian Statistical Institute, 203 B. T. Road, Kolkata 700108, India.}

\author{Nirman Ganguly} 
\email{nirmanganguly@gmail.com}
\thanks{On leave from Department of Mathematics, Heritage Institute of Technology, Kolkata-107, India}
\affiliation{Physics and Applied Mathematics Unit, Indian Statistical Institute, 203 B. T. Road, Kolkata 700108, India.}

\begin{abstract}
The task of device independent secure key distribution requires preparation and subsequently distribution of nonlocal resources. For secured practical implementation, one needs to take two initially uncorrelated quantum systems and perform a unitary on the composite system to generate the nonlocal resource, which is supposed to violate a Bell inequality (say, Bell-CHSH). States which do not violate Bell-CHSH inequality, but violate it when transformed by a global unitary, can be deemed \emph{useful} for the preparation of nonlocal resource. One may then start from a state which is Bell-CHSH local (take for example, a pure product state) and apply an appropriate global unitary on it which results in a Bell-CHSH non-local state. However, an intriguing fact is the existence of \emph{useless} states from which no Bell-CHSH non-local resource can be generated with a global unitary. This is due to the purity preserving nature of unitary operators which bound the amount of correlation in the set of final states depending on the purity of the initial (possibly uncorrelated) states. The present work confirms the existence of such a set, pertaining to two qubit systems. The set exhibits counter intuitive features by containing within it some entangled states which remain Bell-CHSH local on the action of any unitary. From practical perspective, this work draws a line between \emph{useful} and \emph{useless} states for the task of preparing nonlocal resource using global unitary transformations. Furthermore through an analytic characterization of the set we lay down a generic prescription through which one can operationally identify the \emph{useful} states. It has also been shown that our prescription remains valid for any linear Bell inequality.                
\end{abstract}

\maketitle
\section{Introduction}
Entanglement \cite{Review Ent} and nonlocality \cite{Review NL} are considered to be significant resources that quantum mechanics offers. They have a ubiquitous role in information processing tasks and foundational principles, whenever one has to certify quantum advantage over conventional classical procedures. Entanglement is a physical phenomenon in which many particles interact in a manner that the description of each particle separately does not suffice to describe the composite system. A resource of this kind  can be used to demonstrate one of the strongest form of non classical feature i.e non-locality where the statistics generated from each subsystem can not be reproduced by any local realistic theory analogous to classical physics \cite{Bell64,CHSH69}. Bell nonlocal correlation along with entanglement are found to be key resources for many information processing tasks such as teleportation \cite{Bennett93}, dense coding \cite{Bennett92}, randomness certification \cite{random}, key distribution \cite{key}, dimension witness\cite{dw}, Bayesian game theoretic applications \cite{game}. \\

However, the question of identifying an entangled state remains one of the most involved problems in quantum information. Commonly phrased as the \textquote{separability problem}, it has been shown to be NP hard\cite{gurvits1}. In lower dimensions viz. ($2 \otimes 2$ and $2 \otimes 3$) there is an elegant necessary and sufficient criterion criterion to identify entangled states. Negative partial transpose of a quantum state is considered to be a signature of entanglement whereas states having positive partial transpose(PPT) are separable \cite{peres,horodecki}. The solution in higher dimensions lacks a bi-directional logic to certify a state to be entangled, more so with the presence of PPT entangled states \cite{bound}. Nevertheless, an extremely useful operational criteria to detect entanglement is provided through entanglement witnesses(EW)\cite{horodecki,terhal,review}. An outcome of the well-known Hahn-Banach theorem in functional analysis,entanglement witnesses $W$ are hermitian operators having at least one negative eigenvalue which satisfy the inequalities (i) $Tr(W \varrho_{sep}) \geq 0, \forall$ separable states $\varrho_{sep}$ and (ii) $Tr(W\varrho_{ent}) < 0$ for at at least one entangled state $\varrho_{ent}$. The geometric form of the theorem states that points lying outside a convex and closed set can be efficiently separated from the set by a hyperplane  \cite{holmes}. The completeness of this existence guarantees that whenever a state is entangled there is a EW to detect it \cite{horodecki}. On the virtue of being hermitian, entanglement witnesses have proved their efficacy in experimental detection of entanglement \cite{barbieri,wiec}. The notion of this separability axiom has been extended to identify useful resources for teleportation using teleportation witnesses \cite{telwit,telwit2,telwitcom}. An elegant procedure to capture non-locality is through a Bell-CHSH witness \cite{bellwit}. This Bell-CHSH witness is a translation of an EW to detect states which violate the Bell-CHSH inequality.\\
\indent Violation of the Bell-CHSH inequality was first translated in the language of state parameters in \cite{horobell}. For two qubits, all quantum states $\rho$ do not violate the Bell-CHSH inequality iff $M(\rho) \leq 1$, where $M(\rho)$ is defined as the sum of the two largest eigenvalues of the matrix $T_{\rho}^{t}T_{\rho}$, $T_{\rho}$ being the correlation matrix in the Hilbert-Schmidt representation of $\rho$. Thus, $M(\rho) > 1$ is a signature of the non-locality of the state\cite{horobell}. \\ 

The prominent role of non-local states are highlighted through non-local games, device independent quantum key distribution (DI-QKD) and randomness certification. Randomness \cite{random} plays a key role in many information theoretic tasks. It has already been shown to be an important resource for quantum key distribution and cryptography. So an interesting question could be whether one can classify the states which are helpful to certify randomness, that has been answered affirmatively in recent times. Now the question arises whether this class of resources can be expanded in a scenario where prior to the task all the subsystems are subjected to a global unitary operation. In this paper we deal with this question by characterizing the class of local states which can never be \emph{Bell CHSH-Nonlocal} under any possible global unitary operation subsequently termed as absolutely \emph{Bell CHSH-local} states. As one can see that this class of states can never be useful for randomness certification task and DI-QKD. We have further shown that for systems with Hilbert space $\mathbb{C}^2\otimes \mathbb{C}^2$ these states form a convex and compact set. This implies the existence of a hermitian operator which can detect non-absolutely \emph{Bell CHSH-local} states, a potential resource in the modified scenario.\\

In the following section (Sec.\ref{motivate}) we first outline the need for an operator to identify non-absolutely \emph{Bell CHSH-local} states and its importance for a number of information theoretic tasks. The question of existence of a set containing absolutely \emph{Bell CHSH-local} states has also been addressed. In Sec.\ref{def} we introduce the relevant notations and definitions to prepare the required mathematical framework. In Sec.\ref{proof} we present the proof of the existence and a definite scheme of constructing such operators with illustrations of absolutely \emph{Bell CHSH-local} states in Sec.\ref{examples}. Finally we conclude in Sec.\ref{discuss}.    
    
\section{Motivation}\label{motivate}
\begin{figure} 
\includegraphics[scale=0.3]{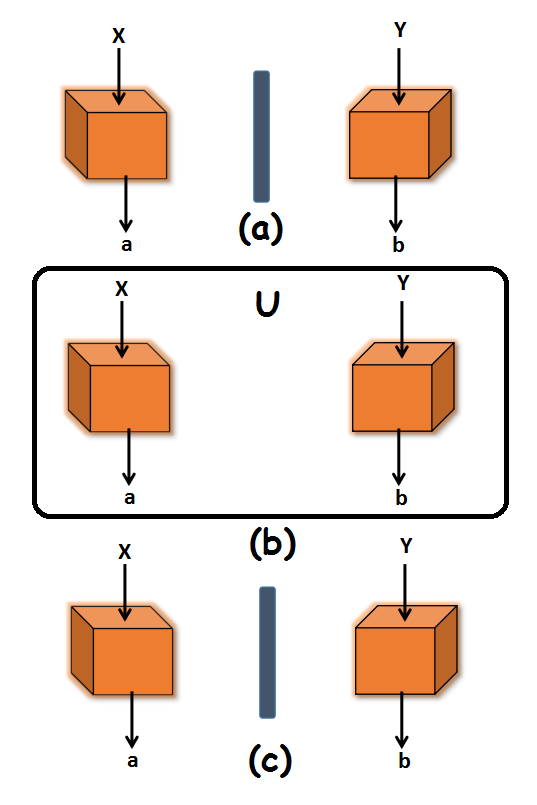} 
\caption{This figure depicts a bipartite scenario where the initially uncorrelated parties (a) are allowed to perform a global unitary operation $\mathbf{U}$ (b). Now the unitary-evolved system is used to play a non-local game (c).}
\label{belldiag}
\end{figure}
Pertaining to separability of quantum states, questions have been raised on the characterization of absolutely separable\cite{vers,johnston} and absolutely PPT states\cite{absppt}. Precisely, a quantum state which is entangled(respectively PPT) in some basis might not be entangled(resp. PPT) in some other basis. This depends on the factorizability of the underlying Hilbert space. Thus, the characterization of states which remain separable(resp. PPT) under any factorization of the basis is pertinent. Literature already contains results in this direction \cite{vers,johnston,absppt,witabs}. Precisely for two qubits , a state is absolutely separable iff its eigenvalues(arranged in descending order) satisfy $\lambda_{1} \leq \lambda_{3} + 2 \sqrt{\lambda_{2} \lambda_{4}}$ \cite{vers}.

In a different perspective, violation of Bell-CHSH inequalities exhibits non-local manifestations of a quantum state. A state which violates the Bell-CHSH inequality is considered non-local. However, the violation of such an inequality invariably depends on the factorization of the underlying Hilbert space and in this regard the study on states which do not violate the Bell-CHSH inequality under any factorization assumes significance. This study forms the main context of our present work.While intuition permits one to state that states which are absolutely separable will be eligible candidates under this classification, it is interesting to probe the existence of entangled states which come under this category. Our results underscore the existence of states outside the absolute separable class , which do not violate Bell-CHSH inequality under any global unitary operation. Global unitary operations play the anchor role here as they transform states to different basis.\\

In standard Bell-scenario the subsystems are allowed to share classical variables prior to the game and perform local operations only. For our purpose we consider a modified Bell-scenario (depicted in Fig.\ref{belldiag}) in which two parties are allowed to perform a global unitary $\mathbf{U}$ prior to the collection of statistics from the joint-system, which we call preparation phase. One can easily note that by performing a CNOT operation on the initial state $(\alpha|0\rangle +\beta|1\rangle)\otimes |0\rangle$ (which is Bell-CHSH local), it can be transformed to $\alpha|00\rangle +\beta|11\rangle$ which has a maximum Bell-CHSH violation of $2\sqrt{1+4|\alpha|^2|\beta|^2}$. This clearly shows that the set of resources (useful initial states) for this modified Bell-CHSH scenario can be expanded to a plausible extent.\\

Now the question is \textit{how far is this extension possible ?}. In this context one might identify the maximally mixed state (the state of minimum purity), as a state which cannot be made non-local by any unitary as it remains same in any basis. However, this is only a trivial insight. Contrary to common intuition there exists states(apart from the maximally mixed state) which do not violate Bell-CHSH inequality under any global unitary operation. We begin by fixing some notations and definitions below. 
  
\section{Notations and Definitions}\label{def} 
$ \mathfrak{B}(X) $ denotes the set of bounded linear operators acting on X. The density matrices that we consider here, are operators acting on two qubits, i.e., $\rho \in \mathfrak{B}(\mathbb{C}^{2} \otimes \mathbb{C}^{2})$. $\mathbf{Q}$ denotes the set of all density matrices.  We denote by $\mathbf{L}$, the set of all states which do not violate the Bell-CHSH inequality \cite{CHSH69}. Recall that any density matrix in two qubits can be written in the Hilbert-Schmidt representation, where $ T_\rho $ denotes the correlation matrix corresponding to $ \rho $. The function $ M(\rho) $  is defined as the sum of the maximum two eigenvalues of $ T_{\rho}^{t}T_{\rho} $. Any state with $ M(\rho) \le 1 $ is considered local with repect to the Bell-CHSH inequality. Hence, the set local set $ \mathbf{L}$ can be expressed mathematically as $\mathbf{L}=\lbrace \rho : M(\rho) \leq 1 \rbrace$. We denote by $ \mathbf{AL} $ as the set containing states which do not violate the Bell-CHSH inequality under any global unitary operation~($ U $) i.e. $\mathbf{AL} = \lbrace \sigma \in \mathbf{L}: M(U \sigma U^{\dagger}) \leq 1 ~ \forall~ U  \rbrace $. One can easily see that $\mathbf{AL}$ forms a non-empty subset of $\mathbf{L}$, as $\frac{1}{4}(I \otimes I) \in \mathbf{AL}$. A schematic diagram of the sets has been shown in Fig.\ref{sets}.
\begin{figure}[b!]
\includegraphics[scale=0.3]{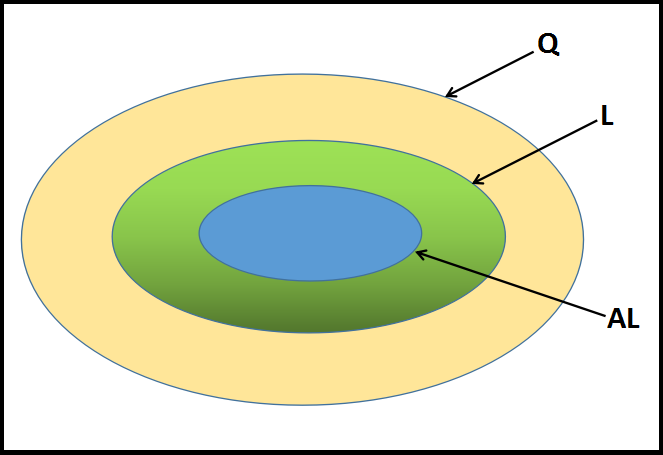}
\caption{A schematic representation of the set $\mathbf{AL}$.}\label{sets}
\end{figure}

\section{Characterization of Absolutely Bell CHSH-local states }\label{proof}
\subsection{Main Result}
The set $ \mathbf{L} $ is characterized by the existence of the Bell-CHSH witness \cite{bellwit} .However, we give a formal characterization below:\\
\begin{theorem}  $\mathbf{L}$ is a convex and compact subset of  $\mathbf{Q}$.
\end{theorem}
\begin{proof}
 First note that the statements below are equivalent:\\
(i) $\rho \in \mathbf{L}$\\
(ii) $ \forall $ Bell-CHSH operator $ B_{CHSH} , Tr(B_{CHSH}\rho) \leq 2$\\
(iii) $ \forall $ Bell-CHSH witness $ B_{CHSH}^{W}, Tr(B_{CHSH}^{W}\rho) \geq 0 $\\
In view of the above, we can rewrite $ \mathbf{L} $ as, $ \mathbf{L} = \lbrace \rho: Tr(B_{CHSH}^{W}\rho) \geq 0, \forall B_{CHSH}^{W}   \rbrace $. Now consider a function $ f_{1}: \mathbf{Q} \rightarrow \mathbb{R} $, defined as
\begin{equation}
f_{1}(\chi)=Tr(B_{CHSH}^{W_{1}}\chi)
\end{equation}
where, $B_{CHSH}^{W_{1}}$ is a fixed Bell-CHSH witness. Let $ L_{1}= \lbrace \chi_{1}:Tr(B_{CHSH}^{W_{1}}\chi_{1}) \geq 0 \rbrace $.  $ Tr(B_{CHSH}^{W_{1}}\chi_{1})  $ will have a maximum value $ d_{1}$ (say). Therefore, one may write $ L_{1} = f_{1}^{-1}[0,d_{1}] $. $ f_{1} $ is a continuous function as $ Tr $ is a continuous function. This in turn implies $ L_{1} $ is a closed set. Continuing as above, one may define $ L_{i} $ for a fixed $B_{CHSH}^{W_{i}}  $. $ L_{i} $ will be closed $ \forall  i $. Since, arbitrary intersection of closed sets is closed, $ \bigcap_{i}L{i} $ is closed. It is easy to see that $ \bigcap_{i}L{i} = \mathbf{L} $ . Hence, $ \mathbf{L} $ is closed. If we now take two arbitrary $ \rho_{1},\rho_{2} \in \mathbf{L}$ , then $ Tr[B_{CHSH}^{W}(\lambda \rho_{1}+(1-\lambda)\rho_{2})] \geq 0$ for any $ B_{CHSH}^{W}  $, $ \lambda \in [0,1] $. This follows from the fact that $ Tr[B_{CHSH}^{W}\rho_{i}] \geq 0 , i=\lbrace 1,2 \rbrace$ for any $ B_{CHSH}^{W} $.Thus $ \mathbf{L} $ is convex. Since $ \mathbf{Q} $ is compact, $ \mathbf{L} $ being a closed subset of $ \mathbf{Q} $, is thus  compact. Hence the theorem.
\end{proof}
This theorem facilitates the characterization of the set $ \mathbf{AL} $ as stated in the theorem below:\\
\begin{theorem} $\mathbf{AL}$ is a convex and compact subset of $\mathbf{L}$.
\end{theorem}
\begin{proof}
We only show that $ \mathbf{AL} $ is convex as the compactness follows from a retrace of the steps presented in \cite{witabs}. \\
Take two arbitrary $ \sigma_{1}, \sigma_{2} \in \mathbf{AL} $. One may rewrite $ \mathbf{AL}= \lbrace \sigma: Tr[B_{CHSH}^{W}(U \sigma U^{\dagger})] \geq 0, \forall B_{CHSH}^{W} , \forall U \rbrace $. Therefore, for any $ U $, $ U[\lambda \sigma_{1} + (1-\lambda) \sigma_{2}]U^{\dagger} = \lambda \sigma_{1}^{\prime} + (1-\lambda)\sigma_{2}^{\prime} \in \mathbf{AL} $. This follows, since $ \mathbf{L} $ is convex. [$ \sigma_{i}^{\prime}=U \sigma_{i} U^{\dagger} $].\\

As noted earlier, one may see the compactness with a re-run of the steps in \cite{witabs}. Hence, the theorem.
\end{proof}
The above  characterization enables to formally define an operator($ W^{B} $) which detects states that violate Bell-CHSH inequality under global unitary. 
\begin{eqnarray}
Tr(W^{B}\sigma) \geq 0 , \forall \sigma \in \mathbf{AL}\label{ineq1}\\
\exists \rho \in \mathbf{L}-\mathbf{AL} , Tr(W^{B}\rho) < 0 \label{ineq2}
\end{eqnarray}
Consider $\rho \in \mathbf{L}-\mathbf{AL}$. There exists a unitary operator $U_{e}$ such that $U_{e}\chi U_{e}^{\dagger}$ violates Bell-CHSH inequality. Consider a 
Bell-CHSH witness $W$\cite{bellwit} that detects $U_{e}\rho U_{e}^{\dagger}$, i.e., $Tr(W U_{e}\rho U_{e}^{\dagger}) < 0$. Using the cyclic property of the trace, 
one  obtains $Tr(U_{e}^{\dagger}W U_{e}\rho ) < 0$. We thus claim that  
\begin{equation}
W^{B}=U_{e}^{\dagger}W U_{e}
\label{witop}
\end{equation}
is our desired operator. To see that it satisfies inequality (\ref{ineq1}), we consider its action on a state $\sigma$ from $\mathbf{AL}$. 
We have $Tr(W^{B}\sigma)=Tr(U_{e}^{\dagger}W U_{e} \sigma)=Tr(W U_{e} \sigma U_{e}^{\dagger} )$. As $\sigma $ $ \in  \mathbf{AL} $,and $ W $ is a Bell-CHSH witness $ Tr(W U_{e} \sigma U_{e}^{\dagger} ) \geq 0$. This implies that $W^{B}$ has a non-negative expectation value on all  states $\sigma \in \mathbf{AL}$. \\
\subsection{Observation regarding the chracterization for any other Bell's inequality}
One may note that the above characterization can be done for any other Bell's inequality, other than the CHSH inequality, with the same run of steps given above.\\
To see that, consider another Bell's inequality given in the form ,
\begin{equation}
Tr(B_{X}\rho_{X}) \leq c
\end{equation}
 where $c$ is any constant and $B_{X}$ is the corresponding Bell operator. 
This can be equivalently expressed in terms of an inequality witness $B_{X}^{W}$ in the form ,
\begin{equation}
Tr(B_{X}^{W}\rho_{X}) \geq 0
\end{equation}
Analogous to the above construction , the local set here can be expressed as ,
$ \mathbf{L_{X}} = \lbrace \rho_{X} : Tr(B_{X}^{W} \rho_{X}) \geq 0 ,\forall B_{X}^{W} \rbrace $
and the absolutely local set as ,
$ \mathbf{AL_{X}} = \lbrace \sigma_{X} : Tr[B_{X}^{W} (U \sigma_{X} U^{\dagger}) ] \geq 0 , \forall B_{X}^{W} , \forall U \rbrace $
where $U$ is unitary. \\
It is easy to verify that a parallel execution of the steps as given for the Bell-CHSH inequality , will characterize $\mathbf{L_{X}}$ and $\mathbf{AL_{X}}$ in an equivalent manner. Consequently,one can construct a witness to detect states living in $\mathbf{L_{X}}-\mathbf{AL_{X}}$, as potential non-local resources under global unitary.
\section{A few examples}\label{examples}
Before we go to some explicit illustrations on absolutely Bell-CHSH local states it is important here to note the changes in the density matrix when it undergoes a global unitary transformation, which has a modifying effect on the \textit{"Bell-CHSH local"} character of the matrix.\\ 
A density matrix in two qubits can be expressed as ,
\begin{equation}
\chi = \frac{1}{4}[I \otimes I + \vec{u}.\vec{s} \otimes I + I \otimes \vec{v}.\vec{s} + \Sigma_{i,j=1}^{3} t_{ij} s_i \otimes s_{j} ]
\end{equation}
Here, $ \vec{u} , \vec{v}  $ are the local bloch vectors and $ t_{ij} = Tr[\chi(s_i \otimes s_j)] $, $ s_{i}  $ are the Pauli matrices.
As noted earlier, in \cite{horobell}, the condition for non-locality was based on the value of a function $ M(\chi) = \lambda_{max1} + \lambda_{max2}; \lambda_{max1}, \lambda_{max2}$ being the maximum two eigen values of $ Y = T^{t}T $ where $ T = [t_{ij}] $ is the correlation matrix of $ \chi $ and $ t $ denotes transposition. A state $ \chi $ is Bell-CHSH local iff $ M(\chi) \le 1 $.

Using Cartan decomposition of unitary matrices, any $ U \in U(4) $ can be decomposed as ,
	\begin{equation}
	U = U_A \otimes U_B U_d V_A \otimes V_B \label{cartan1}
	\end{equation}
	where $ U_A , U_B , V_A , V_B$ are local unitaries and $ U_d $ is the basic non-local unitary \cite{optcreate}. A local unitary $ U_L $ changes the correlation matrix $ T $ corresponding to $ \chi $ in the following way, $ T^ \prime = Q_1 T Q_2^t$, where $ Q_i s$ are rotation matrices with $ det(Q_i)=1 , Q_i^t Q_i=I$ \cite{infhoro}. Therefore, 
	\begin{eqnarray}
	T^ {\prime^t} T^ \prime = Q_2 T^t T Q_2^t
	\end{eqnarray}
Since, the above relation signifies a similarity transformation, the eigenvalues of $ T^t T  $ remain unchanged signalling the invariance of $ M(\chi) $ under local unitary transformation. Therefore, $ M(\chi) $ can only be changed by the action of the basic non-local operator $ U_{d} $. Unlike local unitaries, the basic non-local unitary changes the eigenvalues of $ T^t T$ with contributions from the local bloch vectors. Hence, a global unitary changes the "\textit{Bell-CHSH local}" character of the density matrix.\\
While the theorem in the previous section has provided a tool to identify states which can augment non-local resources under global unitary, it has also highlighted the existence of a set which contains states from which no non-local resource (in terms of the Bell-CHSH inequality) can be generated . It is therefore important to look for certain states which can belong to the absolutely Bell-CHSH local set.

\subsection{Absolutely Separable States}
It is evident that , any separable state obviously belongs to $\mathbf{L}$. From the definition of absolutely separable states, i.e. $\mathbf{AS} = \lbrace \sigma \in \mathbf{S}: U \sigma U^{\dagger} \mbox{is separable,} ~ \forall~ U  \rbrace $, $U$ being any global unitary operation, it is clear that after the operation of the global unitary the state remains in $\mathbf{L}$, i.e. all the absolutely separable states are absolutely Bell-CHSH local states. 

As an illustration ,consider the Bell-diagonal states, i.e. $p_1 |\phi^+\rangle \langle \phi^+|+ p_2 |\phi^-\rangle \langle \phi^-|+p_3 |\psi^+\rangle \langle \psi^+|+p_4 |\psi^-\rangle \langle \psi^-|$, where $\left\lbrace |\phi^+\rangle,|\phi^-\rangle,|\psi^+\rangle,|\psi^-\rangle \right\rbrace$ are the usual Bell states. If we now impose the following two restrictions on the coefficients of the Bell-diagonal states  $(i) \frac{1}{2}>p_1\ge p_2 \ge p_3 \ge p_4, (ii) p_1\le p_3+2\sqrt{p_2 p_4} $, then one notices that they are absolutely separable and hence $\in\mathbf{AL}$.

\subsection{Werner States}
Werner states \cite{Werner} $\sigma_{w}=p |\psi^-\rangle \langle \psi^-| + \frac{1-p}{4} \mathbb{I}$, where $|\psi^-\rangle$ being singlet state, in $2\otimes2$ are absolutely separable for $p\le \frac{1}{3}$, as a result it is also absolutely local here. It can now be asked whether there exist states which are not separable but belong to $\mathbf{AL}$. 

It is well known that $\sigma_{w}$ is entangled but does not violate Bell-CHSH for $p\in (\frac{1}{3},\frac{1}{\sqrt{2}}]$. Now consider that one applies a global unitary operator on it. The change in Bell-CHSH violation will then be contributed by the singlet part only. However since the singlet already exhibits a maximal Bell-violation of $2\sqrt{2}$, the global unitary can only worsen or keep the same Bell-CHSH violation. As a result, the Bell-CHSH violation cannot be maximized with the unitary. Mathematically for any global unitary operation $U$ the first part of the modified state $U|\psi^-\rangle \langle \psi^-|U^\dagger$ will always have Bell-CHSH value less than or equal to $2\sqrt{2}$ and the second part will be zero. 

Therefore it is evident that the Werner states with visibility factor $\frac{1}{\sqrt{2}} \ge p>\frac{1}{3}$ belong to $\mathbf{AL}$ but are entangled. Thus, for visibility factor $p\le \frac{1}{\sqrt{2}}$, the Werner states are absolutely Bell-CHSH local.   

Consider another state: $p|\psi\rangle \langle \psi| + \frac{1-p}{4} \mathbb{I}$, where $|\psi\rangle=a|00\rangle+b|11\rangle$, $a=\frac{1}{\sqrt{3}}$, which is also in $\mathbf{AL}$ for $p\le \frac{1}{\sqrt{2}}$ but entangled for $p> 0.34$.\\

\subsection{Another separable state which is in AL but not in AS}
One may note here that, since unitary similarity is an equivalence relation, it partitions the absolutely Bell-CHSH local set into distinct equivalence classes. For e.g if we define an equivalence class for a definite $\sigma_{al} \in \mathbf{AL}$ as,
\begin{equation}
	[\sigma_{al}]= \lbrace \tau : \tau = U \sigma_{al} U^{\dagger} \rbrace
\end{equation}
then, all states $\tau \in [\sigma_{al}] $ are absolutely Bell-CHSH local.\\
Hence if one considers the state having weights $\left\lbrace \frac{1-p}{4},\frac{1+3p}{4},\frac{1-p}{4},\frac{1-p}{4} \right\rbrace $ diagonal in the computational basis $\left\lbrace |00\rangle ,|01\rangle,|10\rangle,|11\rangle \right\rbrace$, then is a separable state but not absolutely separable for $p>\frac{1}{3}$ and it belongs to $[\sigma_{w}]$ . It is evident from above that the state is in $\mathbf{AL}$ for $p \le\frac{1}{\sqrt{2}}$.

\section{Discussion}\label{discuss}
In standard Bell-scenario the free resources are local operation and shared randomness. Here we have considered a modified scenario where prior to the non-local game the subsystems are allowed to undergo a global unitary evolution. All quantum states can not be made to violate Bell-CHSH inequality even in this modified scenario. In this work we have shown that for two qubit systems these `useless' states which we call absolutely \emph{Bell CHSH-local}, form a convex and compact set implying the existence of a hermitian operator which can detect non-absolutely \emph{Bell CHSH-local} states, a potential resource in the modified Bell-scenario. Here, we would like to mention that in the present work our focus is on the Bell-CHSH violation of a state and not on the corresponding local hidden variable model. However our prescription is not only limited to Bell-CHSH inequality and can be applied to any linear Bell inequality of the form $Tr({\bf B}\rho)\ge 0$. So we can characterize all absolutely local states w.r.t corresponding Bell-inequality. We also present a characterization of absolutely \emph{Bell CHSH-local} states for a number of generic class of states. Global unitary operations change the correlation matrix of a density matrix with contributions from the local bloch vectors. Hence, they change the "Bell-CHSH local" character of a density matrix, which is otherwise impossible with local unitaries. This analysis of Bell non-locality presents a new paradigm for asking a number of important questions. Firstly, one could seek for a generic characterization of absolutely \emph{Bell CHSH-local} states even for two qubits. Secondly, our study leaves open the possibility of the existence of an initial entangled state admitting a LHV model, however violating a Bell inequality when subjected to a global unitary in the preparation phase. Lastly, the question remains whether one could demonstrate the existence of non-absolutely \emph{Bell-local} witness operators for higher dimensional systems in different Bell-scenarios\cite{CGLMP} and subsequently characterizing the set of absolutely \emph{Bell-local} states for such systems.\\
{\bf Acknowledgment:} We would like to gratefully acknowledge fruitful discussions with Prof. Guruprasad Kar. We also thank Tamal Guha and Mir Alimuddin for useful discussions. AM acknowledges support from the CSIR project 09/093(0148)/2012-EMR-I.

\end{document}